\documentclass[11pt]{article}
\usepackage[margin=1in]{geometry}

\newcommand{\SingleSpace}{\edef\baselinestretch{1.00}\Large\normalsize}

\usepackage{epsfig,epic,eepic,units}
\usepackage{hyperref}
\usepackage{url}
\usepackage{longtable}
\usepackage{mathrsfs}
\usepackage{multirow}
\usepackage{bigstrut}
\usepackage{amssymb}
\usepackage{graphicx}
\usepackage{capt-of}
\usepackage{enumerate}	

\usepackage{amsmath}
\usepackage{amsfonts}
\usepackage[T1]{fontenc}

\usepackage{amsthm}
\DeclareMathOperator*{\argmax}{arg\,max}
\newtheorem{theorem}{Theorem}
\newtheorem{corollary}{Corollary}
\newtheorem{lemma}{Lemma}

\newtheorem{definition}{Definition}

\usepackage{algorithm}
\usepackage{algorithmic}

\usepackage{enumitem}
\usepackage{booktabs}
\usepackage{soul}

\newcommand{\specialcell}[2][c]{ \begin{tabular}[#1]{@{}c@{}}#2\end{tabular}}
\usepackage{multirow}
\usepackage{float}
\usepackage[caption = false]{subfig}

\usepackage{authblk}
\title{Fast Distance Sensitivity Oracle for Multiple Failures}
\author[1]{Golshan Golnari}
\author[2]{Zhi-Li Zhang}
\affil[1]{Electrical and Computer Engineering Department, University of Minnesota}
\affil[2]{Computer Science and Engineering Department, University of Minnesota}
\date{}                     
\begin{document}
\SingleSpace

\maketitle

\begin{abstract}
When a network is prone to failures, it is very expensive to compute the shortest paths every time from the scratch. Distance sensitivity oracle provides this privilege to find the new shortest paths faster and with lower cost by once pre-computing an oracle in advance. Although several efficient solutions are proposed in literature to support the single failure, few effort is done to devise an efficient method regarding the case of multiple failures. 

In this paper, we present a novel distance sensitivity oracle based on Markov Tensor Theory \cite{golnari2017markov} to support replacement path queries $(*,t,\mathcal{F})$ in general directed and weighted networks facing the set of failures $\mathcal{F}$. In contrast to the existing work, there is no limitation on maximum failure size supported by our oracle and there is no need to know the size of failure for constructing the oracle. The specifications of our oracle are: space size of $O(n^2)$, pre-process time of $O(n^{\omega})$, where $\omega$ is the exponent of fast matrix multiplication, and query time of $O(m)$ for answering to replacement path query of $(*,t,\mathcal{F})$ which computes the replacement (shortest) paths from all nodes to target $t$ at once. While the computation time for regular shortest path methods, such as Dijkstra's, is $O(m+nlogn)$ for each query after a failure, our algorithm can save a considerable computational time when the size of failure set $|\mathcal{F}|$ is $O(m^{1/\omega})$ or less and the network is sparse $O(m)<O(nlogn)$. 

\end{abstract}




\section{Introduction}

%
%
%
%

Due to the demand for developing more flexible algorithms to support the changes in the network, several problems have been posed under different names and objectives. In the {\em replacement paths} problem, the objective is to answer to query $(s,t,v)$ by computing the shortest replaced path efficiently from a fixed source node $s$ to a fixed target node $t$ for avoiding each of the nodes (or edges) located on the shortest path denoted by $v$. The more general forms of this problem are for multiple sources by answering queries $(*,t,v)$ and the all pairs replacement paths format which answers queries $(*,*,v)$ by efficiently finding the shortest replaced path for all pairs of source and target nodes, while avoiding an arbitrary failed node (or edge) $v$ and by constructing a {\em distance sensitivity oracle}. A more challenging problem is to find the replacement path in case of \textit{multiple failures} and answer to corresponding query $(s,t,\mathcal{F})$, which is still considered as an open problem. The main applications for distance sensitivity oracle are routing in failure-prone networks, Vickrey price problem, and finding $k$ shortest simple paths. When a network is prone to failures, it is very expensive to compute the shortest paths every time from the scratch. Distance sensitivity oracle provides this privilege to compute the new shortest paths faster and with lower cost. In extension, fault-tolerant routing protocol is a distributed solution which seeks for the shortest route avoiding the set of failures while trying to optimize the amount of
memory stored in the routing tables of the nodes (compact routing scheme) \cite{thorup2001compact}. In the Vickrey price problem from auction theory \cite{hershberger2001vickrey} the edges of a networks are each owned by a selfish agent and the objective is to determine the value of an edge according to how difficult it gets to route the information in the network if that edge fails. This can be done by benefiting from sensitivity distance oracle to compare the shortest path length before and after deleting the edge \cite{bernstein2009nearly}. This problem is closely related to find the most damaging or vital node (or edge) in the network \cite{corley1982most}. Moreover, $k$ shortest simple paths can be easily computed by running k executions of a replacement
paths algorithm \cite{eppstein1998finding}.

\textbf{Our Results. }In this paper, we propose a novel and simple-to-implement replacement path algorithm to support \textit{multiple} failures with arbitrary size and answer to $(*,t,\mathcal{F})$ queries efficiently. This algorithm is founded upon two developed concepts in Markov Tensor Theory \cite{Golnari2017MarkovTensor}: avoidance Markov chain and evaporation paradigm. The advantage of our algorithm is multiple folds: 
\begin{enumerate}
	\item By leveraging fast matrix multiplication (with exponent $\omega$, which is currently $\omega=2.376$ \cite{weimann2013replacement}), the sensitivity distance oracle with size $O(n^2)$ is constructed in $O(n^{\omega})$ time. This oracle answers to distance and path queries $(*,t,\mathcal{F})$ in only $O(m+f^{\omega})$ time, where $n$ is the number of nodes, $m$ is the number of edges, and $f$ is the size of failure. We consider the cases with failure size $O(m^{1/{\omega}})$, where the query time becomes even more efficient $O(m)$.
	\item In contrast to the existing work, there is no limitation on maximum failure size supported by our method. In addition, our sensitivity distance oracle does not depend on failure size and can be exploited for any size of the failure once is constructed.
	\item The algorithm supports the general directed networks with arbitrary weights.
	\item  The algorithm can be simply modified to support edge failures.  
	\item The algorithm can find alternative longer paths.	
\end{enumerate}

\textbf{\large{Related Work.}}
Sensitivity distance oracle algorithms have been studied vastly for supporting the single failure case. For weighted and directed networks, Demetrescu et al. \cite{demetrescu2008oracles} proposed an $O(n^2 \log n)$-size oracle which is constructed in $O(mn^2+n^3\log n)$ time and answers to shortest path length queries $(s,t,f)$ in $O(1)$ time. Bernstein and Karger \cite{bernstein2009nearly} improved the previous algorithm by lower construction time of $O(mn\log^2n + n^2\log^3 n)$ but space size of $O(n^2\log^2 n)$ and the same query time of $O(1)$. The same authors also presented a randomized algorithm \cite{bernstein2008improved} which is improved in construction time and storage size with a factor of $\log n$ compared to their deterministic algorithm and the same query time. Note that the query time for \textit{finding} the shortest path is $O(L)$ in all of these algorithms where $L$ is the length of the path. The approximate algorithm proposed by Khanna and Baswana \cite{khanna2010approximate} provides a lower storage requirement of $O(kn^{1+1/k}\frac{log^3 n}{\epsilon^4})$ for unweighted and undirected networks. This algorithm returns $(2k-1)(1+\epsilon)$-approximate distance query in $O(k)$ time for given an integer $k > 1$ and a fraction $\epsilon > 0$.  

As one of the first attempts to support more than one failure, Duan and Pettie \cite{duan2009dual} proposed a method for covering the dual-failures $(f = 2)$. Their method requires the
storage size of $O(n^2 \log^3 n)$ which is constructed in polynomial time. The query time for returning the length of shortest path is $O(\log n)$
and for returning the whole path is $O(L \log n)$. According to the authors, this method cannot be extended to cases with $f > 2$, since it becomes very complex and requires $O(n^f \log^3 n)$ of space. The other $f$-sensitivity distance oracle is a $(8k-2)(f+1)$-approximate algorithm suggested by Chechik et al. \cite{chechik2012f} to support more than two failures $f>2$ for undirected networks. The oracle is constructed in polynomial time and takes $O(fkn^{1+1/k}\log(nW))$ of space to answer distance queries. The query time for this algorithm is $O(|\mathcal{F}|.\log^2n.\log\log n.\log\log L)$, where $\mathcal{F}$ is the number of failures and $\mathcal{F}<f$, $W$ is the weight of heaviest edge, and $L$ is the longest distance in the network.
Weimann and Yuster \cite{weimann2013replacement} propose a randomized algorithm for constructing a sensitivity distance oracle with size of $\tilde{O}(n^{3-\alpha})$ given a trading-off parameter $0 < \alpha < 1$ and conditioned on the failure order being $|\mathcal{F}|=O(\frac{\log n}{\log\log n})$. Notation $\tilde{O}$ indicates that some $\log n$ has been dropped from the order. This algorithm was originally devised for integer-weighted graphs with edge weights chosen from $\{-W,...,+W\}$ \cite{weimann2010replacement} and then was extended to real-weighted graphs in a follow-up work \cite{weimann2013replacement}. For the case of integer weights, the construction time is $O(Wn^{1+\omega-\alpha})$ with query time of $\tilde{O}(n^{2-(1-\alpha)/|\mathcal{F}|})$, and the real weights case has been become possible by construction time of $O(n^{4-\alpha})$ and query time of $\tilde{O}(n^{2-2(1-\alpha)/|\mathcal{F}|})$. The authors take advantage of the fast matrix multiplication, with $\omega$ as the exponent, in their computations which is currently $\omega=2.376$ \cite{weimann2013replacement}.  
In the most recent work, Chechik et al. \cite{chechik20171+} proposed a range of $f$-sensitive distance oracles for undirected networks to answer to queries $(s,t,\mathcal{F})$ conditioned on the failure order being $|\mathcal{F}|=O(\frac{\log n}{\log\log n})$ with $(1+\epsilon)$ stretch. Among the six proposed oracles, excluding the two that are for unweighted networks, the best space requirement is $\tilde{O}(n^2)$ at the cost of $\tilde{O}(n^5)$ pre-process time, and the best pre-process time is $\tilde{O}(n^3)$ at the cost of $\tilde{O}(n^3)$ space requirement. The reviewed works are summarized in Table (\ref{tab:relatedwork}).

\begin{table}[H]
	\begin{center} \label{tab:relatedwork}
		\begin{tabular}{| p{1cm} || c | c | c | c | c | c | }
			\toprule
			\centering
			Ref. & Model & $f$-Max$\#$fails & Approx.& Space&Preprocess&\specialcell[]{Query time\\for $(*,t,\mathcal{F})$}\\ 
			\hline
			\centering
			\cite{duan2009dual}
			&
			\specialcell{directed and\\ weighted}
			& 
			2
			& 
			1
			& 
			$\tilde{O}(n^2)$
			& 
			$poly(n)$
			& 
			$O(nlog n)$
			\\
			\hline
			\centering
\cite{weimann2013replacement}
&
\specialcell{directed and\\ weighted}
& 
$O(\frac{\log n}{\log\log n})$
& 
1
& 
$\tilde{O}(n^{3-\alpha})$
& 
$\tilde{O}(n^{4-\alpha})$
& 
$\tilde{O}(n^{3-2(1-\alpha)/f})$
\\
\hline
			\centering
\cite{chechik20171+}
&
\specialcell{undirected\\and weighted}
& 
$O(\frac{\log n}{\log\log n})$
& 
$1+\epsilon$
& 
$\tilde{O}(n^3)$
& 
$\tilde{O}(mn^2+n^3log n)$
& 
$O(nf^4)$
\\
\hline
			\centering
This paper
&
\specialcell{directed and\\ weighted}
& 
n
& 
1
& 
$O(n^2)$
& 
$O(n^{\omega})$
& 
$O(m+f^{\omega})$
			\\ \bottomrule
		\end{tabular}
		\caption{Related work on distance sensitivity oracle for multiple failures} 
	\end{center}
\end{table}

\section{Method Overview}
The general idea of distance sensitivity oracle presented in this paper is founded on our developed Markov Tensor Theory \cite{golnari2017markov} which is a unified theoretical platform for solving network problems. We have extended the notation of fundamental matrix and hitting time/cost in Markov chain methods to more advanced Markov metrics which we call \textit{avoidance metrics}, such as avoidance fundamental matrix and avoidance hitting time/cost. We take one step further and illustrate the behavior of avoidance metrics in \textit{evaporation paradigm} and show how shortest path information can be nicely resulted from these metrics. 

In the next two sections, we first provide a preliminary review on Markov chain metrics and introduce the avoidance metrics. Then we show how to construct the evaporation paradigm $G_{\alpha}$ from network $G$. We demonstrate that once $\alpha$ goes to 0 in $G_{\alpha}$, the avoidance hitting cost converges to shortest path distance in $G$ and the edges with non-zero probabilities $\c{P}$ represent the edges on the shortest path tree to target $t$. We find the upper bound for $\alpha$ to make this convergence to shortest path happen. Then we illustrate how to devise the distance sensitivity oracle to find the (shortest) replacement paths after multiple failures. 

\section{Preliminaries} \label{sec:prelim}

Consider a weighted and directed network denoted by $G=(V,E,A)$, where $V$ is the set of nodes, $E$ is the set of edges, and $A$ is the adjacency matrix whose $a_{ij}$ entry indicates the distance from $i$ to $j$ if edge $e_{ij} \in E$, otherwise $a_{ij}=0$. A random walk over $G$ is modeled by a Markov chain, where the nodes of $G$ represent the states of the Markov chain and the Markov chain is fully described by its transition probability matrix: $P = D^{-1}A$, where $D=diag[d_i]$ is the diagonal matrix of $d_i$'s, and $d_i=\sum_{j}a_{ij}$ is referred to the (out-)degree of node $i$. In addition, the target nodes in $G$ can be represented as absorbing states in the Markov chain as once being hit, the random walk stops walking around. Throughout the paper, the words ``node" and ``state", and ``network" and ``Markov chain" are used interchangeably.

A Markov chain is called absorbing if it has at least one absorbing state that, once entered, cannot be left. The other states of an absorbing chain, that are not traps, are called non-absorbing or transient states. In an absorbing Markov chain, from each transient state at least one absorbing state should be reachable. Assuming that states are ordered in the way that set of transient
states $\mathcal{T}$ come first and set of absorbing states $\mathcal{A}$ come last, the transition matrix for an absorbing Markov chain takes the following block matrix form:
\begin{equation}
\label{eq:P}
P=\left[
\begin{array}{ c c }
P_{\mathcal{T}\mathcal{T}} & P_{\mathcal{T}\mathcal{A}} \\
0 & I_{\mathcal{A}\mathcal{A}}
\end{array} \right],
\end{equation}
where $I_{\mathcal{A}\mathcal{A}}$ is an $|\mathcal{A}|\times|\mathcal{A}|$ identity matrix and $P$ is row-stochastic.
The fundamental matrix of absorbing chain $P$ is defined as follows:
\begin{equation} \label{eq:N}
F^{\mathcal{A}}=(I-P_{\mathcal{T}\mathcal{T}})^{-1}, 
\end{equation}
where entry $F^{\mathcal{A}}_{sm}$ represents the expected number of passages through state $m$, starting from state $s$, and before absorption by any of absorbing states \cite{snell}. 
To be more clear about the target set $\mathcal{A}$, we show it as a superscript $F^{\mathcal{A}}$.

Expected absorption time, which is also known as (expected) hitting time or first passage time, is calculated as follows: 
\begin{equation} \label{eq:H}
H_s^{\mathcal{A}}=\sum_m F_{sm}^{\mathcal{A}},
\end{equation}
where $H_s^{\mathcal{A}}$ represents the expected number of steps before absorption by {\em any} of the absorbing states in $\mathcal{A}$ when the starting state is $s$.   

To generalize hitting time and account for cost of edges as well, Fouss et al. \cite{fouss2007random} introduced hitting cost metric. Hitting cost for a network with cost matrix $W$ is the average cost incurred by the
random walk when traversing the edges and before hitting the target node for the first time, which can be computed in a recursive form $U_s^{\mathcal{A}}=r_s + \sum_{m\in \mathcal{N}_{out}(s)}P_{sm}U_m^{\mathcal{A}}$ and $r_s$ is the expected out-going cost $r_s=\sum_i p_{si}w_{si}$. We show that the hitting cost can also be computed from fundamental matrix:
\begin{equation} \label{eq:U}
U_s^{\mathcal{A}}=\sum_m F_{sm}^{\mathcal{A}}r_m,
\end{equation}
Notice that hitting
time (\ref{eq:H}) is a special case of hitting cost (\ref{eq:U}) obtained when the cost of all edges are equal to 1, i.e. $w_{ij} = 1$ for all edges $e_{ij}$.

The absorption probability matrix $Q$ is defined as \cite{snell}:
\begin{equation} \label{eq:Q}
Q=FP_{\mathcal{T}\mathcal{A}}
\end{equation}
$Q$ is a $|\mathcal{T}|\times|\mathcal{A}|$ matrix whose $(s,t)$-th entry
is the probability of absorption by absorbing state $t$ when the chain starts from transient state $s$. We denote this entry by $Q_s^{\{t,\overline{\mathcal{F}}\}}$, where $\mathcal{F}=\mathcal{A}\setminus\{t\}$, to be more clear about the absorbing state which is hit first, i.e. $t$ in here, among the other absorbing states which are not touched at all, i.e. $\mathcal{F}$. Note that $\sum_{i\in\mathcal{A}} Q_s^{\{i,\overline{\mathcal{A}\setminus\{i\}}\}}=1$, since starting from any state $s$ it will end up being absorbed by one of the absorbing states eventually.

\section{Avoidance Metrics} \label{sec:avoid}

In this section, we introduce three new Markov chain metrics with modified properties and conditions:

\begin{definition}[Avoidance fundamental matrix]
	Avoidance fundamental matrix for source node $s$, middle node $m$, and target node $t$ conditioned on avoiding node $o$ is computed from classical fundamental matrix and absorption probabilities:
	\begin{equation} \label{eq:avoidanceF}
	F_{s,m}^{\{t,\overline{o}\}}=F_{s,m}^{\{t,o\}}.\frac{Q_m^{\{t,\overline{o}\}}}{Q_s^{\{t,\overline{o}\}}}
	\end{equation}
\end{definition}

\begin{definition}[Avoidance hitting time]
	Avoidance hitting time from $s$ to $t$ avoiding node $o$ is the conditional expectation over the number of steps required to hit $t$ for the first time when starting from $s$ and conditioned on avoiding $o$ on the way, and is obtained from the following equation:
	\begin{equation} \label{eq:avoidanceH}
	H_{s}^{\{t,\overline{o}\}}=\sum_m F_{s,m}^{\{t,o\}}.\frac{Q_m^{\{t,\overline{o}\}}}{Q_s^{\{t,\overline{o}\}}}
	\end{equation}
\end{definition}

\begin{definition}[Avoidance hitting cost]
	Avoidance hitting cost from $s$ to $t$ avoiding node $o$ is the conditional expectation over the cost of steps required to hit $t$ for the first time when starting from $s$ and conditioned on avoiding $o$ on the way, and is obtained from the following equation:
	\begin{equation} \label{eq:avoidanceU}
	U_{s}^{\{t,\overline{o}\}}=\sum_m (F_{s,m}^{\{t,o\}}.\frac{ Q_m^{\{t,\overline{o}\}}}{Q_s^{\{t,\overline{o}\}}})r_m^{\{t,\overline{o}\}}=\sum_m F_{s,m}^{\{t,\overline{o}\}}r_m^{\{t,\overline{o}\}},
	\end{equation}
	where $r_m^{\{t,\overline{o}\}}=\sum_i p_{mi}w_{mi}\frac{Q_i^{\{t,\overline{o}\}}}{Q_m^{\{t,\overline{o}\}}} $.
\end{definition}

In the following, we present a few lemmas which would be required for network analysis applications later on in this paper. 

\begin{lemma}[Incremental Computation of Fundamental Matrix] \label{Nlemma1}
	The fundamental matrix for target set of $\mathcal{S}_1\cup \mathcal{S}_2$ can be computed from the fundamental matrix for target set $\mathcal{S}_1$:
	\begin{equation} 
	F_{im}^{\mathcal{S}_1\cup \mathcal{S}_2}=F_{im}^{\mathcal{S}_1}-F_{i\mathcal{S}_2}^{\mathcal{S}_1}({F_{\mathcal{S}_2\mathcal{S}_2}^{\mathcal{S}_1}})^{-1}F_{\mathcal{S}_2m}^{\mathcal{S}_1},
	\end{equation}
	where the subscripts represent the rows and columns selected from the matrix respectively, e.g. $F_{i\mathcal{S}_2}^{\mathcal{S}_1}$ denotes the row $i$ and the columns corresponding to set $\mathcal{S}_2$ of the fundamental matrix $F^{\mathcal{S}_1}$.
\end{lemma}

\begin{lemma}[Absorption Probability and Normalized Fundamental Matrix] \label{Nlemma3}
	The absorption probability for absorbing set $\{j\}\cup \mathcal{S}$ can be found from the fundamental matrix for absorbing set $\mathcal{S}$:
	\begin{equation} 
	Q_i^{j,\overline{\mathcal{S}}}=\frac{F_{ij}^\mathcal{S}}{F_{jj}^\mathcal{S}}
	\end{equation}
\end{lemma}

For proof of Lemmas (\ref{Nlemma1}) and (\ref{Nlemma3}), please refer to \cite{golnari2017random}.

\begin{corollary}
	Avoidance fundamental matrix $F_{sm}^{t,\overline{k}}$ can be written in terms of classical fundamental matrix $F^k$ by applying Lemmas (\ref{Nlemma1}) and (\ref{Nlemma3}) in the definition of avoidance fundamental matrix (\ref{eq:avoidanceF}): 
	\begin{equation} \label{eq:Favoid_Fo} F_{sm}^{t,\overline{k}}=F_{mt}^k(\frac{F_{sm}^k}{F_{st}^k}-\frac{F_{tm}^k}{F_{tt}^k})
	\end{equation}
\end{corollary}

We present the applications of advanced random walk metrics in next three sections, where the new concept of evaporation paradigm is also being introduced and exploited in conjunction with advanced random walk metrics in the last two ones.

\section{Evaporation Paradigm}

Evaporation paradigm $G_{\alpha}$ is obtained by multiplying factor $\alpha^{w_{ij}}$ into transition probability $P_{ij}$ of $G$ for all edges $\forall e_{ij}\in E$, where $0<\alpha<1$, and adding one (imaginary) node to network, denoted by $o$, to which every other node $i$ is connected with transition probability $1-\sum_{j\in\mathcal{N}(i)} \alpha^{w_{ij}}P_{ij}$. 
\begin{equation} \label{eq:evapP}
P_{ij}(\alpha)=
\begin{cases}
P_{ij}\alpha^{w_{ij}}  & \text{if } i,j\neq o \\
1-\sum_{k\in\mathcal{N}(i)} \alpha^{w_{ik}}P_{ik}  & \text{if } i\neq o \text{ and } j=o \\
0  & \text{if } i= o \text{ and } j\neq o \\
1 & \text{if } i,j= o
\end{cases}
\end{equation}
Thus the new transition probability matrix $P(\alpha)$, belonging to $G_{\alpha}$, is an $(n+1)\times(n+1)$ row-stochastic matrix whose main principal $n\times n$ submatrix is $P_{11}(\alpha)=P\odot\alpha^{W}$, where $\odot$ is the element-wise product. Now with the new transition probability matrix $P(\alpha)$, we compute the avoidance metrics $U_{s}^{t,\overline{o}}(\alpha)$ and $F^{t,\overline{o}}_{sm}(\alpha)$ (from (\ref{eq:avoidanceU}) and (\ref{eq:avoidanceF}) respectively), and generate the routing continuum based on the following theorems.

\begin{theorem} [Routing Continuum: Path Distances]\label{thm:continuum} \label{thm:continuum_weighted}
	Consider weighted network $G$ with at least one path from node $s$ to node $t$. Varying parameter $\alpha$ from 0 to 1 in the avoidance hitting cost of the corresponding evaporating network $G_{\alpha}$ yields a continuum from the shortest-path distance to all-path distance (hitting cost distance) from node $s$ to node $t$ in $G$:
	
	\textit{a}) If $\alpha \rightarrow 0$, $U_{s}^{t,\overline{o}}(\alpha)$ converges to the shortest-path distance from $s$ to $t$ in $G$,
	
	\textit{b}) If $\alpha \rightarrow 1$, $U_{s}^{t,\overline{o}}(\alpha)$ converges to the hitting time distance from $s$ to $t$ in $G$; More precisely, $U_{s}^{t,\overline{o}}(\alpha)$ is exactly equal to the hitting cost distance for $\alpha=1$,
	
	\textit{c}) If $\alpha_1 < \alpha_2$, $U_{s}^{t,\overline{o}}(\alpha_1)\leq U_{s}^{t,\overline{o}}(\alpha_2)$.
	
\end{theorem}

The intuition behind Theorem (\ref{thm:continuum}) is that decreasing $\alpha$, the probability of evaporation in paths increases and when $\alpha$ goes to zero, the probability of longer paths become negligible compared to the probability of the shortest path, and only the shortest path survives. 
In addition, the non-zero entries of matrix $F^{t,\overline{o}}(\alpha)$ become the indicators of the involved nodes lying on the shortest path when $\alpha$ goes to zero, which is demonstrated in the next theorem.

\begin{theorem}[Routing Continuum: Node Flows] \label{thm:F_shp_weighted}
	Consider weighted network $G$ with at least one path from node $s$ to node $t$. For $\alpha \rightarrow 0$ in the corresponding evaporating network $G_{\alpha}$,  the entries of $s$-th row of the avoidance fundamental matrix, i.e. $F^{t,\overline{o}}_{sm}(\alpha)$ for $\forall m \in \mathcal{T}$, determine the following information regarding the shortest path from $s$ to $t$ in network $G$: 
	
	\textit{a}) If $\lim_{\alpha\rightarrow 0} F^{t,\overline{o}}_{sm}(\alpha)=0$, no shortest path from $s$ to $t$ passes through node $m$.
	
	\textit{b}) If $\lim_{\alpha\rightarrow 0} F^{t,\overline{o}}_{sm}(\alpha)=1$, node $m$ is located on all of the shortest paths from $s$ to $t$. 
	
	\textit{c}) If $0<\lim_{\alpha\rightarrow 0}F^{t,\overline{o}}_{sm}(\alpha)<1$, a fraction of the shortest paths from $s$ to $t$ pass through node $m$. 
	
	\textit{d}) As an immediate result of part (c), there exists more than one shortest path from $s$ to $t$ if and only if $\exists m, 0<\lim_{\alpha\rightarrow 0}F^{t,\overline{o}}_{sm}(\alpha)<1$.
\end{theorem}

According to this theorem, computing the $s$-th row of the avoidance fundamental tensor for $\alpha \rightarrow 0$, we can find all of the nodes located on the shortest path(s) from $s$ to $t$. In addition, we can compute the shortest path length $L_{st}$ by summing up over this row (\ref{eq:H}).
In addition, we can find routing continuum edge probabilities (aka how to choose the next edge in a routing) from matrix $\c{P}$ based on the following theorem:

\begin{theorem}[Avoidance Paradigm to Classical Paradigm Transformation] \label{thm:Ptransformation}
	Network $G$ with avoiding node $o$ and target node $t$ can be transformed to network $\c{G}$ without node $o$ and the same target $t$ such that the avoidance metrics in the former network turn into the classical metrics in the latter network, i.e. $F_{sm}^{\{t,\overline{o}\}}=\c{F}_{sm}^{\{t\}}$, $H_s^{\{t,\overline{o}\}}=\c{H}_s^{\{t\}}$, and $U_s^{\{t,\overline{o}\}}=\c{U}_s^{\{t\}}$. The transformation function between transition matrix $\c{P}$ belonging to $\c{G}$ and $P$ belonging to $G$ is as follows:
	\begin{equation} \label{eq:Ptransformation}
	\c{P}_{ij} = P_{ij}\frac{Q_j^{t,\overline{o}}}{Q_i^{t,\overline{o}}}
	\end{equation}
\end{theorem}

\begin{corollary}[Routing Continuum: Edge Probabilities]
	The probabilities assigned to edges for the routing strategy and each choice of $\alpha$ can be obtained from:
	\begin{equation} \label{eq:Ptransformation_alpha}
	\c{P}_{ij}(\alpha) = P_{ij}\alpha^{w_{ij}}\frac{Q_j^{t,\overline{o}}(\alpha)}{Q_i^{t,\overline{o}}(\alpha)}=P_{ij}\alpha^{w_{ij}}\frac{F^o_{jt}(\alpha)}{F^o_{it}(\alpha)},
	\end{equation}
\end{corollary}
where $Q^{t,\overline{o}}(\alpha)$ is computed from (\ref{eq:Q}) and over evaporation transition probability matrix (\ref{eq:evapP}). The second equality is resulted from Lemma (\ref{Nlemma3}).
Algorithm (\ref{alg:shortestpath}) summarizes our method for computing these three metrics to find the continuum information for each choice of $\alpha$.

	\section{Shortest Path} 
	
	According to Theorem (\ref{thm:continuum_weighted}), once $\alpha$ goes to zero, the paths are pruned to \textit{shortest} ones and $U_s^{\{t,\overline{o}\}}$ converges to shortest path distance $L_{st}$. In this section, we prove some theorems for the shortest path case which clarifies the behavior of proposed method for small $\alpha$ and how it can be exploited to devise a novel method for finding the shortest paths.

	\subsection{Avoidance Hitting Cost Convergence Behavior and the Corresponding Error} \label{sec:error}
	
	In this part, we formulate error $\epsilon_{st}(\alpha)=U_s^{\{t,\bar{o}\}}(\alpha)-L_{st}$ in terms of $\alpha$ to study the convergence behavior of avoidance hitting cost to shortest path distance when $\alpha$ goes to 0. This formulation enables us later in this section to find a bound for $\alpha$ to make the error become smaller than $\delta/d$, where $\delta$ is the largest value by which all the edge weights are divisible and $d$ is the out-degree of nodes. Once $\epsilon<\delta/d$, the shortest path distance $L$ can be found by rounding down the avoidance hitting cost $U_s^{\{t,\bar{o}\}}(\alpha)$ to its closest value $k\delta$. We also show that $\epsilon<\delta/d$ is the sufficient condition to find the shortest path from the routing strategy in (\ref{thm:routingstrategy}). 
	
	Let $l_i$'s from countable set $\mathcal{C}$ be the length of walks from $s$ to $t$ such that $L_{st}=l_1<l_2<l_3<...$, and $\textrm{Pr}_{l_i}$'s be the corresponding probabilities (if there are more than one walk with the same length, the $\textrm{Pr}$ is the aggregated probability of the walks). Since $\delta$ is divisible by all walk lengths $l_i$, we can assume that any two consecutive walk lengths differ by $\delta$, i.e. $l_{i+1}=l_i+\delta$, otherwise we can always add a walk length with zero-probability, i.e. $\textrm{Pr}_{l_i}=0$.
	For unweighted networks $\delta=1$. In the evaporating network, every edge $e_{ij}$ is assigned a multiplicative factor of $\alpha^{w_{ij}}$ and so walks with length of $l_i$ have the total probability of $\alpha^{l_i}\textrm{Pr}_{l_i}$. Recall that $0\leq \textrm{Pr}_{l_i}\leq 1$ and $\sum_{l_i} \textrm{Pr}_{l_i}=1$. Then, the avoidance hitting cost can be decomposed into the shortest path distance plus an error term:

	\begin{eqnarray} \label{eq:H=L+res2}
	U_s^{\{t,\overline{o}\}}(\alpha)&=&\frac{\sum_{i=1} {l_i}\alpha^{l_i}\textrm{Pr}_{l_i}}{\sum_{i=1} \alpha^{l_i}\textrm{Pr}_{l_i}} \nonumber
	\\&=& \frac{L_{st}\sum_{i=1} \alpha^{{l_i}}\textrm{Pr}_{l_i}+\sum_{i=2} ({l_i}-l_{i-1})\sum_{k=i} \alpha^{l_k}\textrm{Pr}_{l_k}}{\sum_{i=1} \alpha^{{l_i}}\textrm{Pr}_{l_i}} \nonumber
	\\&=& L_{st}+\delta\frac{\sum_{i=2} \sum_{k=i} \alpha^{l_k}\textrm{Pr}_{l_k}}{\sum_{i=1} \alpha^{{l_i}}\textrm{Pr}_{l_i}} \nonumber
	\\&=& L_{st}+\delta\sum_{j=1}\alpha^{j\delta}\frac{\sum_{i=1} \alpha^{{l_i}}\textrm{Pr}_{l_{i+j}}}{\sum_{i=1} \alpha^{{l_i}}\textrm{Pr}_{l_i}} \nonumber
	\\&=& L_{st}+\delta\sum_{j=1}\alpha^{j\delta}\gamma_j \nonumber
	\\&=& L_{st} + \epsilon_{st}(\alpha), 
	\end{eqnarray}
	
	It can be seen that $\epsilon_{st}(\alpha)$ is a non-negative function of $\alpha$ and so always $U_s^{\{t,\overline{o}\}}(\alpha)\geq L_{st}$, meaning that avoidance hitting cost converges to shortest path distance from above. In the next part, we show that by putting $\epsilon_{st}(\alpha)<\delta/d_s$ and computing the inverse function, a bound for $\alpha$ can be found. 

	\subsection{Finding the Edges on the Shortest Path} \label{sec:findshp}
	
	Beside finding the shortest path distance by computing $U_s^{\{t,\overline{o}\}}(\alpha)$ for small enough $\alpha$, we need to find the path itself. In the following theorem, we show how to find the successor of each node in the shortest path tree and specify the edges located on the shortest path.  
	
	\begin{theorem}[Shortest Path Routing Strategy] \label{thm:routingstrategy}
		Let $\epsilon_{st}(\alpha)<\delta/d_s$, where $d_s$ is the number of out-going neighbors of $s$ and $\delta$ is the largest value by which all the edge weights are divisible. $s$'s out-going edge with highest probability, i.e. $\c{$P$}_{sj}=\max_m \text{\c{$P$}}_{sm}$, is located on the shortest path from $s$ to $t$. 
	\end{theorem}

	Since, finding the shortest path is a recursive process, the whole path can be obtained by finding the successor of each node via the highest edge probability $\c{$P$}_{sj}=\max_m \text{\c{$P$}}_{sm}$ in each step, starting from $s$ and until getting to $t$. 
	The following algorithm summarizes the shortest path routing strategy based on the proposed method

	\begin{algorithm}[H]                       
		\caption{ \sc All pair shortest path }\label{alg:shortestpath}                          
		\begin{algorithmic}
			\small                    
			\STATE {\large\bfseries input:} Probability transition matrix $P$, weight matrix $W$, and $\alpha$ 
			\STATE {\large\bfseries output:} Shortest paths
			\STATE $P(\alpha)=\alpha^W\odot P$
			\STATE $F^o(\alpha)=(I-P(\alpha))^{-1}$
			\FOR {each target $t$}			
			\STATE $\forall e_{ij} \in E: \c{P}^t_{ij}(\alpha)=P_{ij}(\alpha)\frac{F^o_{jt}(\alpha)}{F^o_{it}(\alpha)}$
			\STATE $\forall i \in V: \text{successor$\{i\}$}=\argmax_j \c{P}^t_{ij}(\alpha)$
			\STATE $\text{shortest-path tree rooted at $t$}=\cup_{i\in V}e_{i,\text{successor$\{i\}$}}$
			\ENDFOR  
		\end{algorithmic}
	\end{algorithm}
	
	\subsection{Bound for $\alpha$} \label{sec:bound}
	
	We need a bound for $\alpha$ to make distance error $\epsilon$ smaller than $\delta/d_{max}$ and Theorem \ref{thm:routingstrategy} hold. The following theorem finds such a bound for $\alpha$:
	
	\begin{theorem}
		The distance error can shrink down enough if $\alpha$ follows the following bound:
		\begin{equation} \label{eq:alphaBound}
		\alpha\leq (\frac{1}{(d_{max})^{L_{max}+1}-d_{max}+1})^{1/\delta} \rightarrow \epsilon<\delta/d_{max},
		\end{equation}
		where $L_{max}=\max_{(s,t)} L_{st}$ is the diameter of the network and $d_{max}=\max_i d_i$ is the maximum out-degree in the network 
	\end{theorem}  
	
	\begin{proof}
	We first find an upper bound for $\gamma$ to obtain an upper bound for distance error. Recall that $\epsilon=\delta\sum_{j=1}\alpha^{j\delta}\gamma_j$:
	
	\begin{eqnarray} 
	\gamma_j&=&\frac{\sum_{i=1} \alpha^{l_i}\textrm{Pr}_{l_{i+j}}}{\sum_{i=1} \alpha^{l_i}\textrm{Pr}_{l_i}}\leq \frac{\alpha^{l_1}\sum_{i=1} \textrm{Pr}_{l_{i+j}}}{\sum_{i=1} \alpha^{l_i}\textrm{Pr}_{l_i}} \leq \frac{\alpha^{l_1}(1-(\sum_{i=1}^{j} \textrm{Pr}_{l_i}))}{\alpha^{l_1}\textrm{Pr}_{l_1}} \nonumber
	\\&\leq& \frac{\alpha^{l_1}(1-\textrm{Pr}_{l_1})}{\alpha^{l_1}\textrm{Pr}_{l_1}}=\frac{1-\textrm{Pr}_{l_1}}{\textrm{Pr}_{l_1}}\leq \frac{1-(\frac{1}{d_{max}})^{L_{max}}}{(\frac{1}{d_{max}})^{L_{max}}} 
	\end{eqnarray}
	The last inequality is resulted from the worst case scenario in which the shortest path probability $\textrm{Pr}_{l_1}$ is composed of multiplication of least edge probabilities, i.e. $\frac{1}{d_{max}}$, and for the longest distance of network diameter. The upper bound for distance error $\epsilon_{st}$ is obtained as follows:
	\begin{equation} \label{eq:eps_upperbound}
	\epsilon_{st} \leq \delta\sum_{i=1}\alpha^{i\delta}(\frac{1-(\frac{1}{d_{max}})^{L_{max}}}{(\frac{1}{d_{max}})^{L_{max}}})=\delta\frac{\alpha^{\delta}}{1-\alpha^{\delta}}(\frac{1-(\frac{1}{d_{max}})^{L_{max}}}{(\frac{1}{d_{max}})^{L_{max}}})
	\end{equation}
	
	To guarantee that the distance error is smaller than $\delta/d_{max}$, we can make its upper bound (\ref{eq:eps_upperbound}) be lower than $\delta/d_{max}$, i.e. $\delta\frac{\alpha^{\delta}}{1-\alpha^{\delta}}(\frac{1-(\frac{1}{d_{max}})^{L_{max}}}{(\frac{1}{d_{max}})^{L_{max}}})<\frac{\delta}{d_{max}}$. Now, we can find a bound for $\alpha$ in terms of $\delta$, network diameter $L_{max}$, and maximum out-degree $d_{max}$ to have distance error $\epsilon$ smaller than $\delta/d_{max}$:

	\begin{equation}
	\alpha\leq (\frac{1}{(d_{max})^{L_{max}+1}-d_{max}+1})^{1/\delta}\approx (\frac{1}{d_{max}})^{(L_{max}+1)/{\delta}}
	\end{equation}
	
\end{proof}

\section{Replacement Path After Multiple Failures} \label{sec:theory}

\begin{theorem} \label{thm:replacepath}
	Assume set $\mathcal{F}$ of nodes have failed in weighted network $G$. If $\alpha\rightarrow 0$ in the corresponding evaporating network $G_{\alpha}$, the avoidance hitting cost $U_s^{t,\{\overline{\mathcal{F},o}\}}(\alpha)$ in $G_{\alpha}$ converges to shortest-path distance in $G$ where failure set $\mathcal{F}$ is discarded from the network.
	\begin{equation}
	\lim_{\alpha \rightarrow 0} U_s^{t,\{\overline{\mathcal{F},o}\}}(\alpha)=L_{st}^{\overline{\mathcal{F}}}
	\end{equation}
\end{theorem}

Algorithm (\ref{alg:replacepath}) pre-computes the distance sensitivity oracle $F^o(\alpha)$ and answers to replacement shortest paths queries efficiently from all nodes to target $t$ while there are a set of failure nodes $\mathcal{F}$ in the network.

\begin{algorithm}[H]                       
	\caption{ \sc Replacement path algorithm for all sources to single target queries with multiple failures $(*,t,\mathcal{F})$ }\label{alg:replacepath}                          
	\begin{algorithmic}
		\small
		\STATE {\large\bfseries Input:} 
		\STATE Probability transition matrix $P$, weight matrix $W$, and $\alpha$                     
		\STATE {\large\bfseries Output:} 
		\STATE Shortest paths from all nodes to single target $t$ which do not pass any nodes in failure set $\mathcal{F}$
		\STATE {\large\bfseries Preprocess:}
		\STATE $P(\alpha)=\alpha^W\odot P$
		\STATE $F^o(\alpha)=(I-P(\alpha))^{-1}$
		\STATE {\large\bfseries Query:}  $(*,t,\mathcal{F})$ 
		\STATE {\large\bfseries Query response:}
		\STATE $M=(F_{\mathcal{F},\mathcal{F}}^o(\alpha))^{-1}$
		\STATE $\forall i \in V: F_{i,t}^{\mathcal{F},o}(\alpha)=F_{i,t}^o(\alpha)-F_{i,\mathcal{F}}^o(\alpha)MF_{\mathcal{F},t}^o(\alpha)$			
		\STATE $\forall e_{ij} \in E: \c{P}^t_{ij}(\alpha)=P_{ij}(\alpha)\frac{F^{\mathcal{F},o}_{jt}(\alpha)}{F^{\mathcal{F},o}_{it}(\alpha)}$
		\STATE $\forall i \in V: \text{successor$\{i\}$}=\argmax_j \c{P}^t_{ij}(\alpha)$
		\STATE $\text{Shortest-path tree rooted at $t$}=\cup_{i\in V}e_{i,\text{successor$\{i\}$}}$
	\end{algorithmic}
\end{algorithm}
where the second equation in Query response is resulted from Theorem (\ref{Nlemma1}) and the third one is a substitution of (\ref{Nlemma3}) in Theorem (\ref{thm:routingstrategy}).

\subsection{Preprocess time and space} \label{sec:complexity}
The purpose of preprocess part is to compute and store $F^o(\alpha)$ which can be used to answer replacement path queries very efficiently. The required space for storing this matrix is $n^2$ where $n$ is the number of nodes. Regarding the complexity time, the inverse computation is the main costly component with complexity time of $O(n^{\omega})$, where $\omega=2.376$, and is discussed in the following. Recall that the ReccomenderModule requires 20 network metrics as input who are all less complex than $O(n^{\omega})$.

\textbf{Matrix Inverse: }
The computational complexity of matrix multiplication of two $n\times n$ matrices is sub-qubic; according to Strassen algorithm \cite{strassen1969gaussian} the complexity is $O(n^{2.807})$ and later on it reduced even more to $O(n^{2.376})$ by Coppersmith-Winograd algorithm \cite{coppersmith1997rectangular}. Cormen et al. \cite{cormen2001introduction} proved that inversion is no harder than multiplication (Theorem 28.2). A divide and conquer algorithm that uses blockwise inversion to invert a matrix runs with the same time complexity as the matrix multiplication algorithm that is used internally.

\subsection{Query time}

For having a fast query time, we leverage the incremental computation in Theorem (\ref{Nlemma1}). Based on this theorem, only an $O(f^{\omega})$-computation is required to compute $F_{i,t}^{\mathcal{F},o}(\alpha)$ from precomputed matrix $F^o(\alpha)$ and for a given failure set $\mathcal{F}$ with size $f$. The other most costly component of query computations is computing the new probabilities $\c{P}^t_{ij}(\alpha)$ for all edges which takes $O(m)$ time and makes the overall computational time of $O(m+f^{\omega})$ for each query. By considering the cases with failure size $f<m^{1/{\omega}}$, the overall query time reduces to $O(m)$


\section*{Acknowledgement}
The research was supported in part by  US DoD DTRA grants HDTRA1-09-1-0050
and HDTRA1-14-1-0040, and ARO MURI Award W911NF-12-1-0385.

\newpage
\bibliographystyle{unsrt} 
\bibliography{main}

\newpage
\section{Appendix}
\subsection{Proof of Theorems} \label{sec:rplthmProof}

\begin{proof}[Proof of Theorem \ref{thm:continuum_weighted}]
	Let $l_i$'s from countable set $\mathcal{C}$ be the length of walks from $s$ to $t$ such that $L_{st}=l_1<l_2<l_3<...$, and $\textrm{Pr}_{l_i}$'s be the corresponding probabilities, where $\sum_{i=1} \textrm{Pr}_{l_i}=1$. The avoidance hitting cost (\ref{eq:general_avoidanceU}) in evaporation network finds the following form:
	\begin{equation} 
	U_s^{\{t,\overline{o}\}}(\alpha)=\frac{\sum_{i=1} l_i\alpha^{l_i}\textrm{Pr}_{l_i}}{\sum_{i=1} \alpha^{l_i}\textrm{Pr}_{l_i}},
	\end{equation}

	\textit{Proof of part (a)}
	
	When $\alpha \rightarrow 0$, the first term of numerator (and denominator) which is for $l_1=L_{st}$ dominates the subsequent terms and $U_s^{\{t,\overline{o}\}}(\alpha)$ converges to $\frac{ L_{st}\alpha^{L_{st}}\textrm{Pr}_{L_{st}}}{\alpha^{L_{st}}\textrm{Pr}_{L_{st}}}=L_{st}$.

	\textit{Proof of part (b)}
	
	For $\alpha = 1$, there is no evaporation and network $G_{\alpha}$ splits into two disconnected subgraphs: the original network $G$ with node $t$ as its only absorbing node, and one isolated node which is node $o$. Then $U_s^{\{t,\overline{o}\}}(\alpha)$ reduces to the regular hitting cost from $s$ to $t$ in the original network $G$:
	\begin{equation}
	U_s^{\{t,\overline{o}\}}(\alpha=1)=\frac{\sum_{i=1} l_i\textrm{Pr}_{l_i}}{\sum_{i=1} \textrm{Pr}_{l_i}}=\sum_{i=1} l_i\textrm{Pr}_{l_i}=U_s^t
	\end{equation}
	
	\textit{Proof of part (c)}
	
	We prove that if $\alpha_1 < \alpha_2$ then $U_s^{\{t,\overline{o}\}}(\alpha_1)\leq U_s^{\{t,\overline{o}\}}(\alpha_2)$, i.e.:
	\begin{equation} \label{ineq_a_weighted}
	\frac{\sum_{i=1} l_i\alpha_1^{l_i}\textrm{Pr}_{l_i}}{\sum_{i=1} \alpha_1^{l_i}\textrm{Pr}_{l_i}} \leq \frac{\sum_{i=1} l_i\alpha_2^{l_i}\textrm{Pr}_{l_i}}{\sum_{i=1} \alpha_2^{l_i}\textrm{Pr}_{l_i}}
	\end{equation}

	Cross-multiplying the fractions in (\ref{ineq_a_weighted}), we compare the corresponding terms from the left-hand-side and right-hand-side polynomials. Without loss of generality assume that $i\leq j$:
	\begin{eqnarray}
	(\alpha_2^{l_i} \textrm{Pr}_{l_i})({l_j}\alpha_1^{l_j} \textrm{Pr}_{l_j})+(\alpha_2^{l_j} \textrm{Pr}_{l_j})({l_i}\alpha_1^{l_i} \textrm{Pr}_{l_i}) &\leq&  (\alpha_2^{l_i} \textrm{Pr}_{l_i})({l_j}\alpha_1^{l_j} \textrm{Pr}_{l_j})+(\alpha_2^{l_j} \textrm{Pr}_{l_j})({l_i}\alpha_1^{l_i} \textrm{Pr}_{l_i})  \nonumber \\ 
	\Rightarrow \textrm{Pr}_{l_i}\textrm{Pr}_{l_j}({l_j}\alpha_2^{l_i}\alpha_1^{l_j}+{l_i}\alpha_2^{l_j}\alpha_1^{l_i}) &\leq&  \textrm{Pr}_{l_i}\textrm{Pr}_{l_j}({l_j}\alpha_1^{l_i}\alpha_2^{l_j}+{l_i}\alpha_1^{l_j}\alpha_2^{l_i}) \nonumber
	\end{eqnarray}
	Notice that for this inequality in two cases of: 1) $\textrm{Pr}_{l_i}$ or $\textrm{Pr}_{l_j}$ being zero, and 2) $i=j$ the equality holds; otherwise:
	\begin{equation}
	({l_j}-{l_i})\alpha_2^{l_i}\alpha_1^{l_j} < ({l_j}-{l_i})\alpha_1^{l_i}\alpha_2^{l_j} \qquad \Longrightarrow \quad \alpha_1^{{l_j}-{l_i}} < \alpha_2^{{l_j}-{l_i}},\nonumber
	\end{equation}
	where the last inequality is obviously correct.	
\end{proof}

\begin{proof} [Proof of Theorem \ref{thm:F_shp_weighted}]
	The avoidance fundamental matrix in evaporation network when the network is weighted finds the following form:
	\begin{equation}
	F_{sm}^{\{t,\overline{o}\}}(\alpha)=\frac{(\sum_{l_i=L_{sm}}\alpha^{l_i}\sum_{\zeta_j\in Z_{sm}(l_i)}\textrm{Pr}_{\zeta_j})\cdot(\sum_{l_i=L_{mt}}\alpha^{l_i}\sum_{\zeta_j\in Z_{mt}(l_i)}\textrm{Pr}_{\zeta_j})}{\sum_{l_i=L_{st}}\alpha^{l_i}\sum_{\zeta_j\in Z_{st}(l_i)}\textrm{Pr}_{\zeta_j}}
	\end{equation}
	
	When $\alpha \rightarrow 0$, the first terms with lowest exponent of $\alpha$ dominate the subsequent terms and the equation above reduces to:
	\begin{equation} \label{eq:limit_weighted}
	\lim_{\alpha \rightarrow 0} F_{sm}^{\{t,\overline{o}\}}(\alpha)=\lim_{\alpha \rightarrow 0}\frac{\alpha^{L_{sm}+L_{mt}}(\sum_{\zeta_j\in Z_{sm}(L_{sm})}\textrm{Pr}_{\zeta_j}).(\sum_{\zeta_j\in Z_{mt}(L_{mt})}\textrm{Pr}_{\zeta_j})}{\alpha^{L_{st}} \sum_{\zeta_j\in Z_{st}(L_{st})}\textrm{Pr}_{\zeta_j}}
	\end{equation}
	
	\textit{Proof of part (a)}
	\\If $m$ is not located on any shortest path from $s$ to $t$, then $\alpha^{L_{sm}+L_{mt}}>\alpha^{L_{st}}$ and the limit in Eq. (\ref{eq:limit_weighted}) converges to zero.
	
	\textit{Proof of part (b)\&(c)}
	\\If $m$ is located on at least one of the shortest paths from $s$ to $t$, then $\alpha^{L_{sm}+L_{mt}}=\alpha^{L_{st}}$ and the limit (\ref{eq:limit_weighted}) has non-zero value: $\lim_{\alpha \rightarrow 0} F_{sm}^{\{t,\overline{o}\}}(\alpha)>0$. On the other hand, we know that $\sum_{\zeta_j\in Z_{st}(L_{st})}\textrm{Pr}_{\zeta_j}\geq (\sum_{\zeta_j\in Z_{sm}(L_{sm})}\textrm{Pr}_{\zeta_j})\cdot(\sum_{\zeta_j\in Z_{mt}(L_{mt})}\textrm{Pr}_{\zeta_j})$ if $L_{sm}+L_{mt}=L_{st}$. In the case that $m$ is located on all of the shortest paths from $s$ to $t$, it should be in $L_{sm}$ distance from $s$ and in $L_{mt}$ distance to $t$ on all of these paths (otherwise we can find a shorter path by connecting two shorter pieces) and thus we have: $\sum_{\zeta_j\in Z_{st}(L_{st})}\textrm{Pr}_{\zeta_j}= (\sum_{\zeta_j\in Z_{sm}(L_{sm})}\textrm{Pr}_{\zeta_j})\cdot(\sum_{\zeta_j\in Z_{mt}(L_{mt})}\textrm{Pr}_{\zeta_j})$ which results to $\lim_{\alpha \rightarrow 0} F_{sm}^{\{t,\overline{o}\}}(\alpha)=1$. However, if $m$ is not located on all of the shortest paths from $s$ to $t$, then we have $\sum_{\zeta_j\in Z_{st}(L_{st})}\textrm{Pr}_{\zeta_j}> (\sum_{\zeta_j\in Z_{sm}(L_{sm})}\textrm{Pr}_{\zeta_j})\cdot(\sum_{\zeta_j\in Z_{mt}(L_{mt})}\textrm{Pr}_{\zeta_j})$ and so $\lim_{\alpha \rightarrow 0} F_{sm}^{\{t,\overline{o}\}}(\alpha)<1$.

\end{proof}

\begin{proof}[Proof of Theorem \ref{thm:Ptransformation}]
	We first prove that $\c{P}$ is a transition probability matrix, namely is row stochastic:
	\begin{equation}
	\sum_{j\in\mathcal{N}(i)} \c{P}_{ij} = \sum_{j\in\mathcal{N}(i)} P_{ij}\frac{Q_j^{T,\overline{o}}}{Q_i^{T,\overline{o}}} =  \frac{1}{Q_i^{T,\overline{o}}}\sum_{j\in\mathcal{N}(i)} P_{ij}Q_j^{T,\overline{o}}=\frac{Q_i^{T,\overline{o}}}{Q_i^{T,\overline{o}}}=1,
	\end{equation}
	where the third equality is resulted because of $Q$ is a harmonic function. Now we show that with the transformation in eq. (\ref{eq:Ptransformation}) these equalities hold:  $F_{sm}^{\{T,\overline{o}\}}=\c{F}_{sm}^{\{T\}}$, $H_s^{\{T,\overline{o}\}}=\c{H}_s^{\{T\}}$, and $U_s^{\{T,\overline{o}\}}=\c{U}_s^{\{T\}}$.
	\begin{eqnarray}
	\c{$F$}^{\{T\}}&=&\sum_{k=0} \c{P}_{\mathcal{TT}}^k=\sum_{k=0} (Diag(Q^{T,\overline{o}})^{-1}P_{\mathcal{TT}}Diag(Q^{T,\overline{o}}))^k \nonumber
	\\&=&\sum_{k=0} Diag(Q^{T,\overline{o}})^{-1}P_{\mathcal{TT}}^k Diag(Q^{T,\overline{o}}) \nonumber
	\\&=& Diag(Q^{T,\overline{o}})^{-1}(\sum_{k=0}P_{\mathcal{TT}}^k) Diag(Q^{T,\overline{o}}) \nonumber
	\\&=&Diag(Q^{T,\overline{o}})^{-1}F^{\{T,o\}} Diag(Q^{T,\overline{o}}) \nonumber
	\\&=&F^{\{T,\overline{o}\}} \nonumber
	\end{eqnarray}
	For the hitting times we have $\c{H}_s^{\{T\}}=\c{$F$}^{\{T\}}\textbf{1}$ and $H_s^{\{T,\overline{o}\}}=F^{\{T,\overline{o}\}} \textbf{1}$, so $H_s^{\{T,\overline{o}\}}=\c{H}_s^{\{T\}}$.
	The following relations also hold for hitting costs:
	\begin{eqnarray}
	U_s^{\{T,\overline{o}\}}&=&\sum_m F_{sm}^{\{T,\overline{o}\}}r_m^{\{T,\overline{o}\}} \nonumber
	\\&=&\sum_m F_{sm}^{\{T,\overline{o}\}}\sum_j \frac{Q_j^{\{T,\overline{o}\}}}{Q_m^{\{T,\overline{o}\}}}P_{mj}w_{mj} \nonumber
	\\&=&\sum_m F_{sm}^{\{T,\overline{o}\}}\sum_j \c{$P$}_{mj}w_{mj} \nonumber
	\\&=&\sum_m F_{sm}^{\{T,\overline{o}\}} \c{$r$}_m \nonumber
	\\&=&\sum_m \c{$F$}_{sm}^{\{T\}}\c{$r$}_m \nonumber
	\\&=&\c{$U$}_s^{\{T\}}, \nonumber
	\end{eqnarray}
	where the first and third equalities are based on (\ref{eq:avoidanceU}) and (\ref{eq:Ptransformation}) respectively. 
\end{proof}

\begin{proof} [Proof of Theorem \ref{thm:routingstrategy}]
	We first find an expression for distance error $\epsilon_{st}(\alpha)$ in terms of edge costs and probabilities. 
	According to Theorem (\ref{thm:Ptransformation}), any avoidance paradigm can be transformed to a corresponding classical paradigm, and we have: $	\c{$P$}_{ij}(\alpha) = P_{ij}(\alpha)\frac{Q_j^{t,\overline{o}}(\alpha)}{Q_i^{t,\overline{o}}(\alpha)}$, $	\c{$F$}_{sm}^{\{t\}}(\alpha)=F_{sm}^{\{t,\overline{o}\}}(\alpha)$, and $\c{$U$}_s^{\{t\}}(\alpha)=U_s^{\{t,\overline{o}\}}(\alpha)$.
	In the transformed classical paradigm, we can write the recursive function of hitting cost (\ref{sec:prelim}) and transform it back to corresponding avoidance metrics. (Just note that for the rest of the proof, we drop $\alpha$'s to avoid clutter and make the relations more readable):
	
	\begin{eqnarray}
	\c{$U$}_s^{\{t\}}&=&\text{\c{r}}_s + \sum_{m\in \mathcal{N}_{out}(s)}\c{P}_{sm}\c{U}_m^{\{t\}} \nonumber
	\\\rightarrow  U_s^{\{t,\overline{o}\}}&=&\text{\c{r}}_s + \sum_{m\in \mathcal{N}_{out}(s)}\c{P}_{sm}U_m^{\{t,\overline{o}\}} \nonumber
	\\&=& \sum_{m\in \mathcal{N}_{out}(s)}\c{P}_{sm}w_{sm}+ \sum_{m\in \mathcal{N}_{out}(s)}\c{P}_{sm}U_m^{\{t,\overline{o}\}} \nonumber
	\\&=&\c{P}_{sj}(w_{sj}+U_j^{\{t,\overline{o}\}})+\sum_{m\neq j}\c{P}_{sm}(U_m^{\{t,\overline{o}\}}+w_{sm}) \nonumber
	\\\rightarrow L_{st}+\epsilon_{st}&=& \c{P}_{sj}(w_{sj}+L_{jt}+\epsilon_{jt})+\sum_{m\neq j}\c{P}_{sm}(L_{mt}+\epsilon_{mt}+w_{sm}) \nonumber
	\end{eqnarray}
	
	Out-going edge set of node $s$ can be divided into two subset of $\mathcal{J}_e$ and $\mathcal{J}^C_e$, where $\mathcal{J}_e$ consists of the edges that are located on the shortest path from $s$ to $t$, and $\mathcal{J}^C_e$ is the complementary set. Let $\mathcal{J}_v$ be the corresponding out-going neighbors to $\mathcal{J}_e$ , i.e. $\mathcal{J}_e=\cup_{i\in\mathcal{J}_v} e_{si}$ and $|\mathcal{J}_e|=|\mathcal{J}_v|$. We prove that the edge with highest probability $\c{$P$}_{sj}=\max_m \text{\c{$P$}}_{sm}$ belong to $\mathcal{J}_e$. If $\mathcal{J}_e$ includes all of $s$'s out-going edges and $\mathcal{J}^C_e=\emptyset$, the proof is complete; otherwise, there exists at least one $s$'s out-going edge which is not located on the shortest path from $s$ to $t$, i.e. $|\mathcal{J}_e|\leq d_s-1$. Now, we show that the maximum edge probability in set $\mathcal{J}_e$ is higher than the maximum edge probability in $\mathcal{J}^C_e$:
	\begin{eqnarray}
	\rightarrow\epsilon_{st}&=& (\sum_{j\in \mathcal{J}_v}\c{P}_{sj}-1)L_{st}+\sum_{j\in \mathcal{J}_v}\c{P}_{sj}\epsilon_{jt}+\sum_{m\notin \mathcal{J}_v}\c{P}_{sm}(L_{mt}+\epsilon_{mt}+w_{sm}) \nonumber
	\\&\geq& (\sum_{j\in \mathcal{J}_v}\c{P}_{sj}-1)L_{st}+\sum_{m\notin \mathcal{J}_v}\c{P}_{sm}(L_{mt}+w_{sm}) \nonumber
	\\&\geq& (\sum_{j\in \mathcal{J}_v}\c{P}_{sj}-1)L_{st}+\sum_{m\notin \mathcal{J}_v}\c{P}_{sm}(L_{st}+\delta) \nonumber
	\\&=& (\sum_{j\in \mathcal{J}_v}\c{P}_{sj}-1)L_{st}+(1-\sum_{j\in \mathcal{J}_v}\c{P}_{sj})(L_{st}+\delta) \nonumber
	\\&=& (1-\sum_{j\in \mathcal{J}_v}\c{P}_{sj})\delta \label{eq:epsLowerBound}
	\end{eqnarray}
	Substituting the lower bound of $\epsilon_{st}$ (\ref{eq:epsLowerBound}) in the Theorem's assumption of $\epsilon_{st}<\delta/d_s$, the following result is obtained:
	\begin{equation}
	(1-\sum_{j\in \mathcal{J}_v}\c{P}_{sj})\delta<\delta/d_s
	\rightarrow\sum_{j\in \mathcal{J}_v}\c{P}_{sj}> \frac{d_s-1}{d_s}\rightarrow\sum_{j\in \mathcal{J}^C_v}\c{P}_{sj}< \frac{1}{d_s},
	\end{equation}
	which proves that the highest edge probability in $\mathcal{J}_e$ is at least equal to $\frac{1}{d_s}$, while the highest edge probability in $\mathcal{J}^C_e$ is strictly less than $\frac{1}{d_s}$.
	
\end{proof}

\begin{proof}[Proof of Theorem \ref{thm:replacepath}] \label{prf:repPath_unweighted}
	Let $G$ be an unweighted network and avoidance hitting time $H_s^{t,\{\overline{\mathcal{F},o}\}}(\alpha)$ is defined on the corresponding evaporation paradigm $G_{\alpha}$, and $\mathcal{F}$ is the set of failure nodes. We write the avoidance hitting time in terms of transition probabilities (\ref{eq:general_avoidanceH}):
	\begin{equation}
	H_s^{t,\{\overline{\mathcal{F},o}\}}(\alpha)=\frac{\sum_{k=k_1} k[{[P(\alpha)]}^{k-1}_{\mathcal{T}_2\mathcal{T}_2}{[P(\alpha)]}_{\mathcal{T}_2\mathcal{A}_2}]_{st}}{\sum_{k=k_1} [{[P(\alpha)]}^{k-1}_{\mathcal{T}_2\mathcal{T}_2}{[P(\alpha)]}_{\mathcal{T}_2\mathcal{A}_2}]_{st}}=\frac{\sum_{k=k_1} k\alpha^k[P^{k-1}_{\mathcal{T}_1\mathcal{T}_1}P_{\mathcal{T}_1\mathcal{A}_1}]_{st}}{\sum_{k=k_1} \alpha^k[P^{k-1}_{\mathcal{T}_1\mathcal{T}_1}P_{\mathcal{T}_1\mathcal{A}_1}]_{st}},
	\end{equation}
	where $P(\alpha)$ is the transition matrix of evaporation network and $P$ belongs to the original network. In the original network $G$ the target node $t$ as well as the failure set $\mathcal{F}$ form the absorbing set: $\mathcal{A}_1=\{t\}\cup \mathcal{F}$ and $\mathcal{T}_1=V\setminus\mathcal{A}_1$. In the evaporating network $G_{\alpha}$, the evaporation node $o$ is absorbing too: $\mathcal{A}_2=\{o\}\cup \mathcal{A}_1$ and $\mathcal{T}_2=V\setminus\mathcal{A}_2$. When $\alpha\rightarrow 0$
	\begin{equation}
	\lim_{\alpha \rightarrow 0} H_s^{t,\{\overline{\mathcal{F},o}\}}(\alpha)=\lim_{\alpha \rightarrow 0}\frac{\sum_{k=k_1} k\alpha^k[P^{k-1}_{\mathcal{T}_1\mathcal{T}_1}P_{\mathcal{T}_1\mathcal{A}_1}]_{st}}{\sum_{k=k_1} \alpha^k[P^{k-1}_{\mathcal{T}_1\mathcal{T}_1}P_{\mathcal{T}_1\mathcal{A}_1}]_{st}}=k_1,
	\end{equation}
	$k_1$ is the smallest number of steps to take from $s$ to reach $t$ in the transient part of $G$, which interprets the shortest path distance from $s$ to $t$ excluding the nodes in $\mathcal{F}$. 
	
	For the weighted network, the proof is straightforward following the same idea for the unweighted network as well as using Theorem (\ref{thm:continuum_weighted}).
\end{proof}

\subsection{Network Example}
This network example shows how varying $\alpha$ from 1 to 0, the edges are pruned to the ones located on the shortest path tree. This phenomena is described through main indicators of paths, i.e. $\c{P}^t(\alpha)$, $F^{t,\overline{o}}(\alpha)$, $U^{t,\overline{o}}(\alpha)$, which are computed for five different values of $\alpha$ and target $t=6$, and presented in Table (\ref{tbl:continuum}). The routing strategy in terms of edge probability $\c{P}^t(\alpha)$ for these five different values of $\alpha$ are depicted in Fig. (\ref{fig:continuum}b-f).
\begin{figure}[H]
	\vskip -0.1in
	\begin{center}
		\subfloat[directed weighted network example (weights on the edges)]{\label{fig:continuum_network}\includegraphics[width=0.3\textwidth]{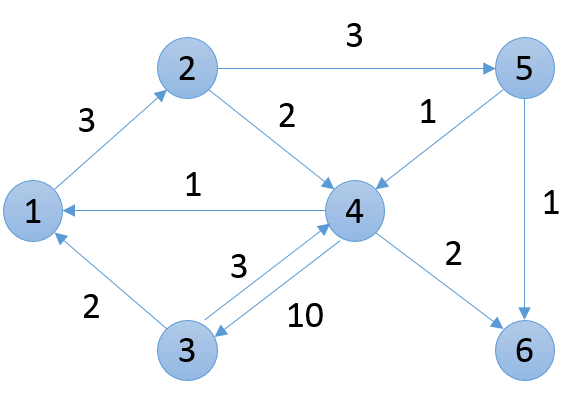}} ~
		\subfloat[$\alpha=0.0001$]{\label{fig:network_P_a00001}\includegraphics[width=0.3\textwidth]{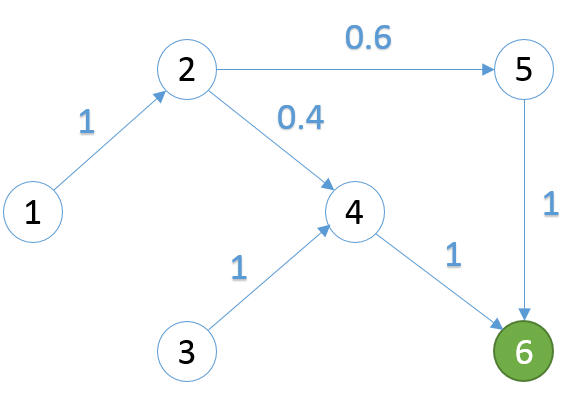}} ~
		\subfloat[$\alpha=0.3$]{\label{fig:network_P_a03}\includegraphics[width=0.3\textwidth]{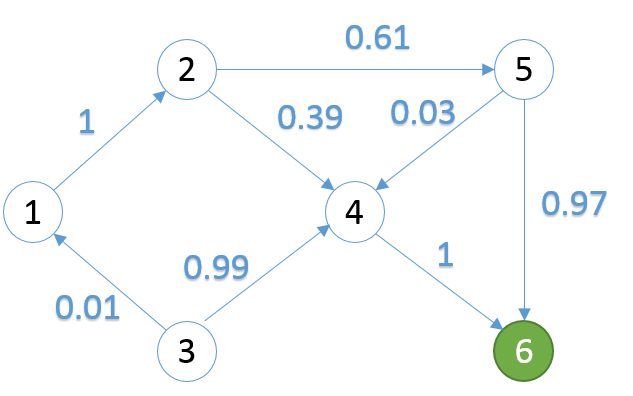}}
		\\
		\subfloat[$\alpha=0.6$]{\label{fig:network_P_a06}\includegraphics[width=0.3\textwidth]{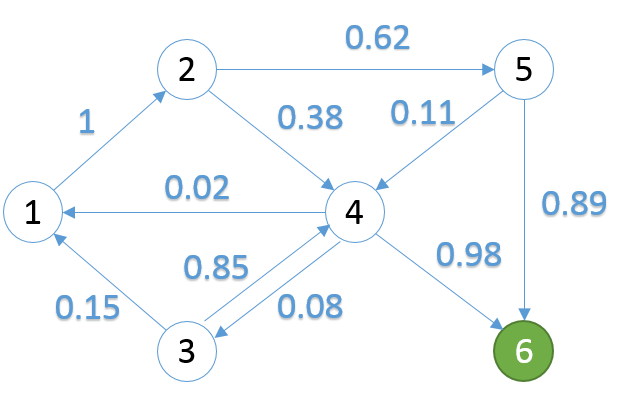}} ~
		\subfloat[$\alpha=0.9$]{\label{fig:network_P_a09}\includegraphics[width=0.3\textwidth]{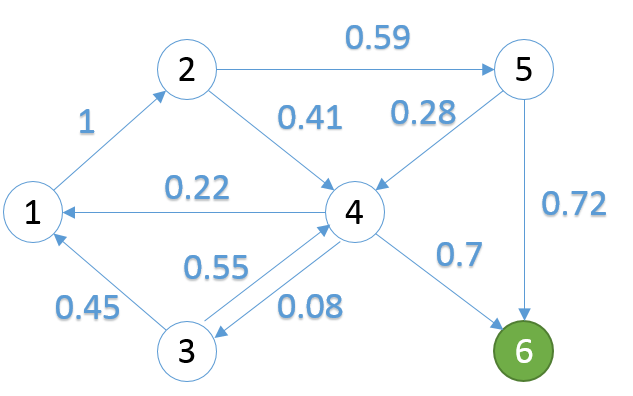}} ~
		\subfloat[$\alpha=1$]{\label{fig:network_P_a1}\includegraphics[width=0.3\textwidth]{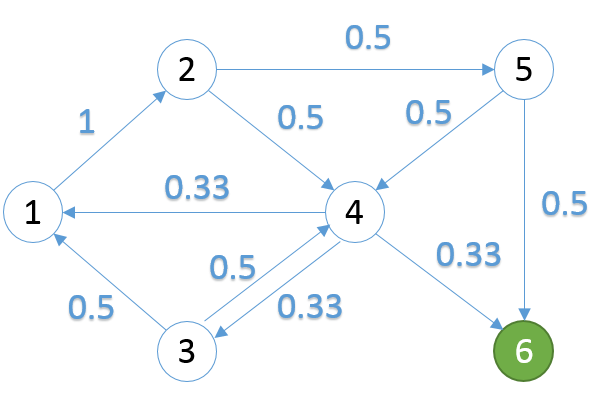}}
		\caption{\small Routing continuum: (b)-(f) show routing edge probabilities for network example in (a) and for different values of $\alpha$ which generate a continuum from shortest path to all path. The weights on the edges in (a) represent the cost of edges and in (b)-(f) indicates the routing edge probabilities. Target is node 6.}
		\label{fig:continuum}
	\end{center}
	\vskip -0.15in
\end{figure}

\begin{table}[H]
	\begin{center}
		\begin{tabular}{| p{1cm} || c | c | c | }
			\toprule
			\centering
			$\alpha$ & $\c{P}^{\{6\}}(\alpha)$ & $F^{\{6,\overline{o}\}}(\alpha)$ & $U^{\{6,\overline{o}\}}(\alpha)$ \\ 
			\cmidrule(r){1-1}\cmidrule(lr){2-2}\cmidrule(lr){3-3}\cmidrule(l){4-4}
			\vspace{1cm}
			\centering
			$0.0001$
			&
			\raisebox{-\totalheight}{\includegraphics[width=0.4\textwidth]{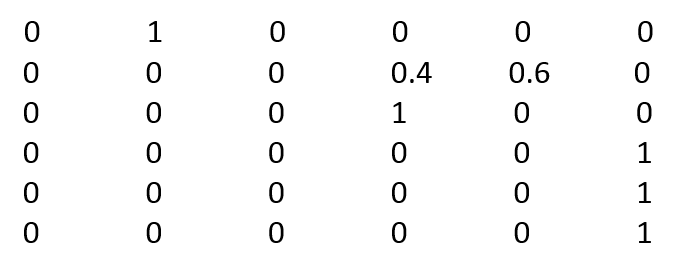}}
			& 
			\raisebox{-\totalheight}{\includegraphics[width=0.4\textwidth]{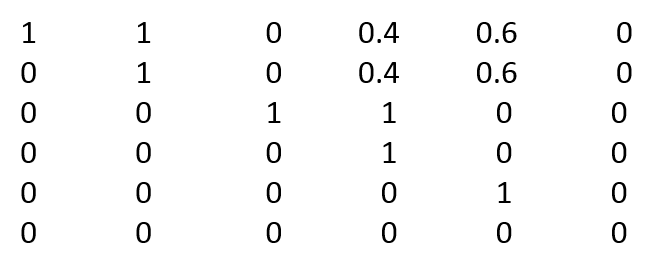}}
			& 
			\raisebox{-\totalheight}{\includegraphics[width=0.068\textwidth]{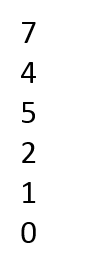}}
			\\
			\hline
			\vspace{1cm}
			\centering
			$0.3$
			&
			\raisebox{-\totalheight}{\includegraphics[width=0.4\textwidth]{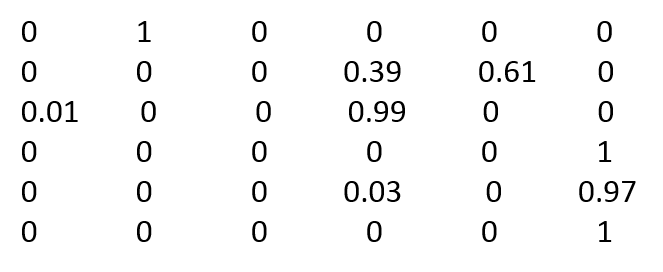}}
			& 
			\raisebox{-\totalheight}{\includegraphics[width=0.4\textwidth]{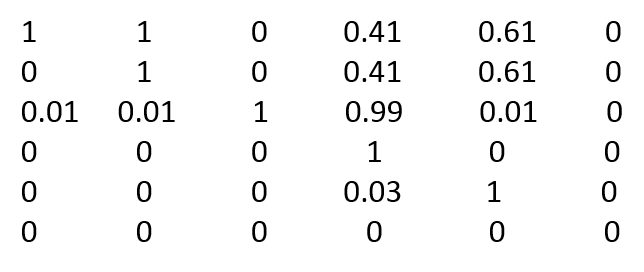}}
			& 
			\raisebox{-\totalheight}{\includegraphics[width=0.068\textwidth]{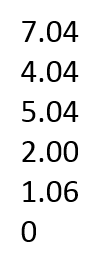}}
			\\
			\hline
			\vspace{1cm}
			\centering
			$0.6$
			&
			\raisebox{-\totalheight}{\includegraphics[width=0.4\textwidth]{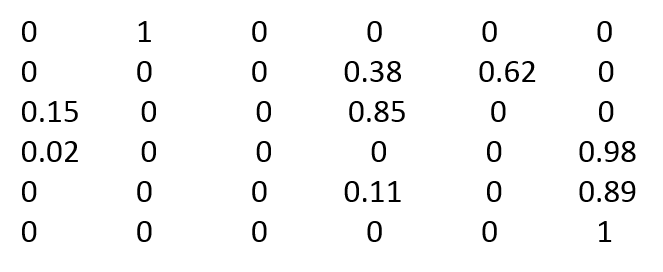}}
			& 
			\raisebox{-\totalheight}{\includegraphics[width=0.4\textwidth]{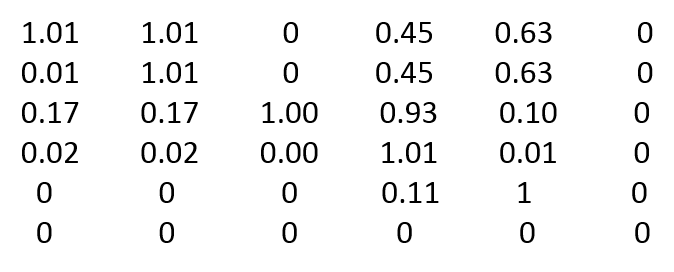}}
			& 
			\raisebox{-\totalheight}{\includegraphics[width=0.068\textwidth]{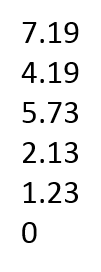}}
			\\
			\hline
			\vspace{1cm}
			\centering
			$0.9$
			&
			\raisebox{-\totalheight}{\includegraphics[width=0.4\textwidth]{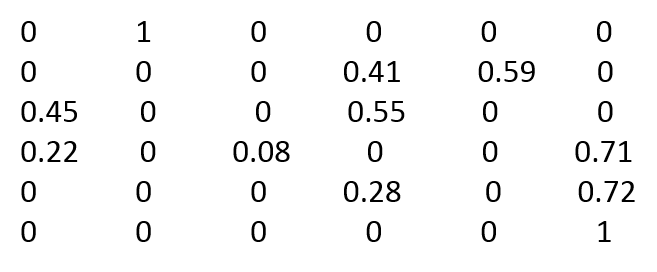}}
			& 
			\raisebox{-\totalheight}{\includegraphics[width=0.4\textwidth]{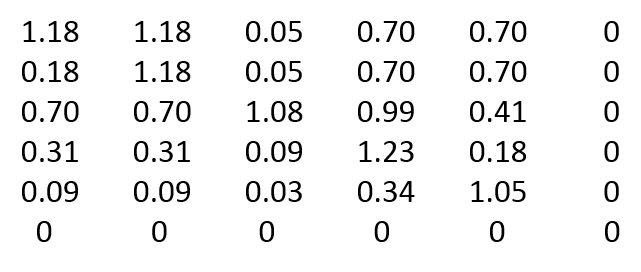}}
			& 
			\raisebox{-\totalheight}{\includegraphics[width=0.068\textwidth]{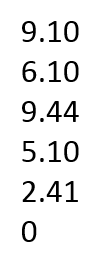}}
			\\
			\hline
			\vspace{1cm}
			\centering
			$1$
			&
			\raisebox{-\totalheight}{\includegraphics[width=0.4\textwidth]{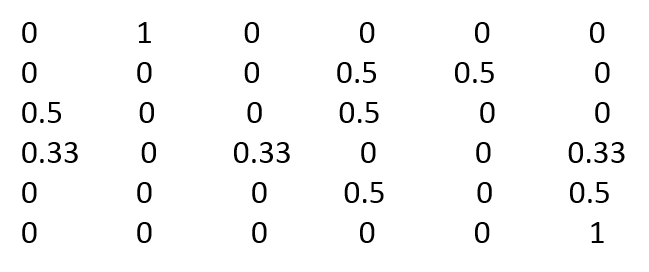}}
			& 
			\raisebox{-\totalheight}{\includegraphics[width=0.4\textwidth]{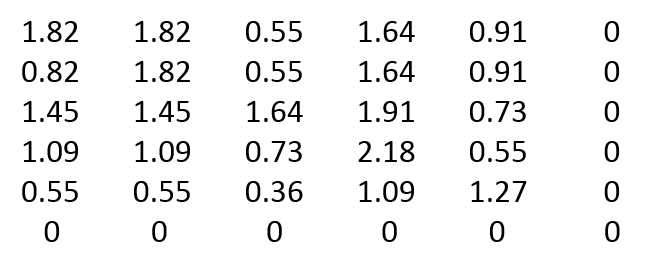}}
			& 
			\raisebox{-\totalheight}{\includegraphics[width=0.072\textwidth]{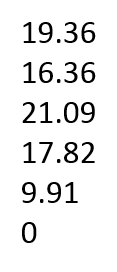}}
			\\ \bottomrule
		\end{tabular}
		\caption{Continuum path indicators for target node $t=6$ and different choices of $\alpha$ for network example in Fig. (\ref{fig:continuum_network})}
		\label{tbl:continuum}
	\end{center}
\end{table}

$U^{6,\overline{o}}(\alpha)$ indicates the vector of distances from all nodes to node 6. It can be seen for $\alpha$ close to zero ($\alpha=0.0001$ in Table (\ref{tbl:continuum})) these distances are the same as shortest path distances. For larger $\alpha$'s these distances grow monotonically until $\alpha=1$ where the classical hitting cost distances are obtained $U^{6,\overline{o}}(\alpha=1)=H^6$.

$F^{6,\overline{o}}(\alpha)$ represents the stochastic flow of nodes to target node 6. It is specially meaningful for two extreme cases of $\alpha=0.0001$ and $\alpha=1$; e.g. $F^{6,\overline{o}}_{1j}(\alpha=0.0001)$ indicates the stochastic portion of shortest paths from node 1 to 6 that pass through node j, which is 0.4 for $j=4$, 0.6 for $j=5$, 1 for $j=2$ which implies that all of the shortest paths from 1 to 6 pass through node 2, and 0 for $j=3$ indicating no shortest path from 1 to 6 pass through node 3. Existence of any value larger than 0 and smaller than 1 in $i$-th row of $F^{t,\overline{o}}(\alpha\rightarrow 0)$ indicate the existence of \textit{multiple} shortest paths from $i$ to $t$. 
\\For the other extreme $\alpha=1$, $F^{6,\overline{o}}(\alpha=1)$ is representing the expected visit times in regular random walks, i.e. classical fundamental matrix $F^6$.

$\c{P}^6(\alpha)$ is the matrix of edge probabilities for routing purposes. In other words, when a packet arrives at node $i$ it is forwarded over edge $e_{ij}$ with probability $\c{P}^6_{ij}(\alpha)$. Thus $\c{P}^6_{ij}(\alpha)$ indicates the usage portion of edge $e_{ij}$ for routing packets from $i$ to $t=6$ and for parameter $\alpha$. For $\alpha=0.0001$ (shortest path case), it can be seen that edges not belonging to shortest paths have zero probability (and so not shown in the figure), and the non-zero-probability edges form the shortest DAG from all the nodes to target node 6 (Fig. (\ref{fig:network_P_a00001})).


%
%
%
\end{document}